\documentclass[11pt, oneside]{article} 
\usepackage[margin=2cm]{geometry}  
\usepackage{graphicx}
\usepackage{newtxtext}
\usepackage{newtxmath}
\usepackage[authoryear]{natbib}
\usepackage{hyperref}
\hypersetup{
    colorlinks = true,
    urlcolor   = blue,
    citecolor  = black,
}
\newcommand{\RomanNumeralCaps}[1]
\linenumbers
\usepackage{graphicx}%
\usepackage{dcolumn}%
\usepackage{bm}%
\usepackage{hyperref}%
\usepackage{graphicx}%
\usepackage{dcolumn}%
\usepackage{bm}%
\usepackage[english]{babel}
\usepackage{amsmath}

\usepackage{amsthm}
\usepackage{wasysym}
\usepackage{float} 
\usepackage[color=pink]{todonotes}
\usepackage{subcaption}
\usepackage{comment}
\usepackage{framed}

\usepackage{listings}%

\newtheorem{theorem}{Theorem}[section]
\newtheorem{remark}{Remark}[section]

\newtheorem{lemma}{Lemma}[section]

\newtheorem{proposition}{Proposition}[section]

\newcommand{\p}{\partial}

\begin{document}
\title{Compound Burgers -- KdV Soliton Behaviour: Refraction, Reflection and Fusion}

\author{Darryl D. Holm$^1$, Ruiao Hu$^1$, Oliver D. Street$^{1,2}$, Hanchun Wang$^{1,3}$\footnote{Corresponding author, email: hw660@cam.ac.uk} \\ \footnotesize
(1) Department of Mathematics, Imperial College London. \\ \footnotesize
(2) Grantham Institute, Imperial College London. \\ \footnotesize
(3) Department of Applied Mathematics and Theoretical Physics, University of Cambridge.
\\ \small
}
\date{}
\maketitle
\begin{abstract}

\noindent We consider a coupled PDE system between the Burgers equation and the KdV equation to model the interactions between `bore'-like structures and wave-like solitons in shallow water. Two derivations of the resulting Burgers–swept KdV system are presented, based on Lie group symmetry and reduced variational principles. Exact compound soliton solutions are obtained, and numerical simulations show that the Burgers and KdV momenta tend toward a balance at which the coupled system reduces to the integrable Gardner equation. The numerical simulations also reveal rich nonlinear solution behaviours that include refraction, reflection, and soliton fusion, before the balance is finally achieved.

\begin{description}
\item[Key Words: ] 
 Burgers Equation, KdV Equation, soliton, near-integrable system
\end{description}
\end{abstract}
{
  \hypersetup{hidelinks}
  \tableofcontents
}

\section{\label{sec:level1}Introduction}

The derivation of nonlinear models of Wave Current Interaction (WCI) dynamics has a long history in fluid mechanics. Standard treatments include the Generalised Lagrangian Mean theory by \cite{AM1978} as extended, e.g., in \cite{Grimshaw1984}. For recent reviews, see e.g., \cite{Ablowitz2011, Buhler2014}. 

One of the most dramatic and well-studied examples of WCI arises in the propagation of a  tidal bore. A tidal bore is a strong tide that pushes up a river, against the current. In the frame of the front, the waves behind the front move backward and the waves ahead of the front move forward. These undulating waves are known as `whelps'. 

\color{black}
During the past few decades, a new paradigm for modelling WCI has emerged which treats the elevation of surface waves as a type of \emph{symmetry breaking} of the surface associated with the current flow. In the symmetry breaking paradigm, surface wave elevation fields are modelled as \emph{order parameters} whose dynamics evolve as a separate degree of freedom while being swept along by the fluid flow. E.g., surface waves modelled as localised structures through various generalisations of the Korteweg–de Vries (KdV) equation in \citep{ostrovsky2024localized}.

\begin{figure}[H]
\centering
\includegraphics[width=0.9\linewidth]{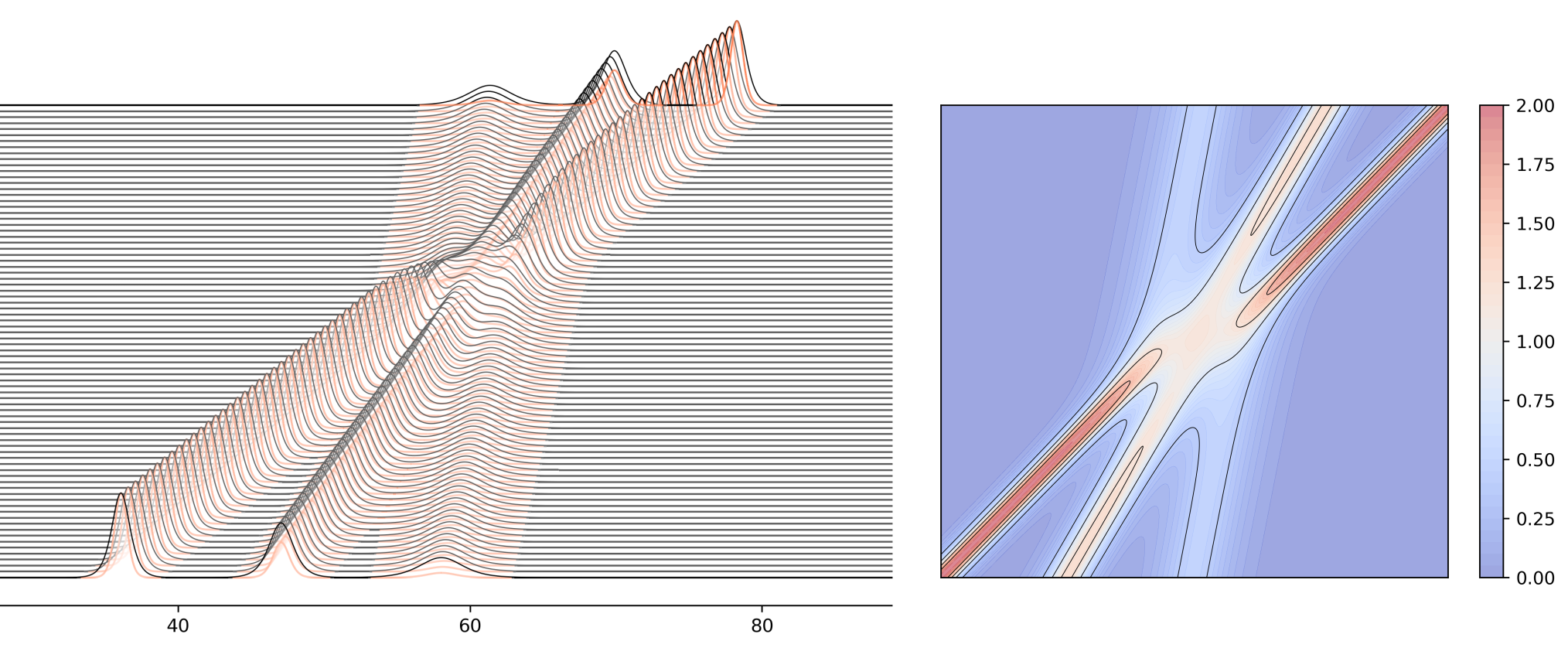}
    \caption{Propagation and collisions of three compound solitons in the Burgers–swept KdV system \eqref{eq: Burgers–swept KdV} that satisfies \( u = v^2/2 \). Left: waterfall plot of \( u \) (red, Burgers velocity) and \( v \) (black, KdV velocity). Right: contour plot of \( v(x,t) \), showing the emergence and collision of three compound solitons, which retain their forms but reverse order after interaction.}
    \label{Gardner-fig}
\end{figure}

Following this symmetry-breaking paradigm, the Burgers-swept Korteweg-de Vries (Burgers–swept KdV) system \eqref{eq: Burgers–swept KdV} was first proposed in \cite{Dombret2025}. In this formulation, the Euler-Poincar\'e Lagrangian of the Burgers equation is coupled to the Lagrangian of the KdV equation following the Dirac--Frenkel (DF) methodology 
\citep{Frenkel1934WaveMA}. 
\begin{equation}
\begin{split}
\textbf{Burgers: \qquad}&u_t + 3u{u_x} \qquad \;\;\;\;\; =  - v{\partial _x}(  3{v^2} + {\gamma{v_{xx}}} )\,,\\
\textbf{KdV: \qquad}&v_t +  {{6v{v_x} + \gamma v_{xxx}}}  =  - {\partial _x}(uv)
\,.\end{split}
\label {eq: Burgers–swept KdV}
\end{equation}
In the DF approach, one couples the Lagrangians in Hamilton's principle for the two separate degrees of freedom by introducing the Doppler shift which boosts the KdV wave momentum into the frame of the Burgers flow velocity.
\smallskip

In the equivalent Burgers–swept KdV system,
\begin{equation}
 {\begin{aligned}
\textbf{Gardner: \qquad} &{{v_t} + 6{vv_x} + \tfrac{3}{2}{v^2v_x} + \gamma v_{xxx} = -\,\p_x (vm)}\,,\\
\textbf{Lie Transport: \qquad} &m_t + (m\partial_x + \partial_x m) u =  0 
\,,\quad \hbox{with}\quad m := u - \tfrac12 v^2
\,,\end{aligned}} 
\label{v+ep-eqn-intro}
\end{equation}
the Gardner equation for velocity, $v$, is forced by a quantity involving momentum $m :=u-\tfrac12 v^2$ and, in turn, the quantity $m$ is Lie transported by the Burgers flow velocity, $u$.
This is the sense in which the KdV wave dynamics is `swept' in the frame of the Burgers bore. 

If initially $m(x,0)=0$ were to hold throughout the entire flow domain, then homogeneity of the transport equation in \eqref{v+ep-eqn-intro} would ensure that $m$ remains zero. In this case, the Burgers–swept KdV system would reduce to the integrable Gardner's equation for $v(x,t)$ and would exhibit the compound soliton dynamics shown in Figure \ref{Gardner-fig}.

\smallskip

However, if $m \ne 0$ initially, then solutions of the Burgers–swept KdV system may explore the rich class of nonlinear wave--current interactions beyond the integrable regime which are discussed in the remainder of the paper. 
\smallskip

Numerical simulation in section \ref{sec-3} shows how the Burgers–swept KdV system in \eqref{eq: Burgers–swept KdV} introduces compound soliton interactions and wave-current dynamics that go well beyond the behaviour observed in simpler models. Namely, a Burgers-type bore may merge with incoming compound solitons and amplify them. This nonlinear interaction leads to the formation of a new type of compound wave-current structure, in which the wave is propagating in the frame of the current, whose flow is nonlinearly influenced by the presence of the wave. Depending on the interaction dynamics, these compound solutions may either be reflected (evolving into faster waves that propagate away from the bore), or refracted (penetrating into the bore and propagating as part of its internal structure). A threshold condition is identified that distinguishes between refraction and reflection regimes in overtaking collisions of Burgers bores with KdV solitons. Fusion events are observed in which multiple solitons merge into fewer outgoing stronger compound Gardner solitons upon exceeding the refraction threshold. These findings extend the classical Gardner framework to a broader class of WCI phenomena.

\textbf{Content of the paper.}\\
In this paper, we move beyond qualitative analysis and undertake a detailed investigation of the nonlinear dynamics of the Burgers–swept KdV system \eqref{eq: Burgers–swept KdV}. Our main contributions are to: \\

1. Establish a connection between the Burgers–swept KdV system and the well-known completely integrable Gardner equation, demonstrating that the latter can be recovered as a special case of the former.\\
2. Derive the compound soliton solutions in the Burgers–swept KdV system.\\
3. Provide three distinct travelling wave solutions with a zero-speed steady profile.\\
4. Investigate wave-current interactions at the front of the Burgers bore through simulation, exploring phenomena such as reflection, refraction, and fusion.\\

\section{Properties of the Burgers-swept KdV System}\label{sec-2}
\noindent In this section, we show that the Burgers–swept KdV system \eqref{eq: Burgers–swept KdV}, modeling wave–current interactions, can be derived 
via the Dirac–Frenkel variational method for particles and fields. We then provide three classes of solutions: compound solitons, solitary waves, and periodic waves.

\subsection{Variational derivation of the Burgers–swept KdV system 
}
Taking the approach studied in \cite{HOLM2023133847}, we couple the systems by adding Burgers and KdV Lagrangians together and Doppler shifting the time derivative within the KdV contribution by the velocity $u$. This is reminiscent of the Doppler shifting of the dispersion relation studied by \cite{Pizzo_Salmon_2021}. We consider the following action integral
\begin{equation}
\begin{aligned}
	S_\text{WCI} &= \int_{t_0}^{t_1}\int_{\mathbb{R}} \frac{u^2}{2} + \frac12 v(\p_t + {u\p_x}) (\mathcal{D}^{-1}v) - v^3 +  \frac{\gamma}{2} v_x^2\,dx\,dt \,,
\end{aligned}
\label{eq: WCI_lagrangian}
\end{equation}
with the Doppler shifting $\partial_t \rightarrow (\partial_t + u\partial_x)$ and \( \mathcal{D} = -\partial_x \) as in the KdV Lagrangian
\footnote{
\noindent The KdV equation has a noncanonical Hamiltonian structure. It can be written as \( \partial_t v = \mathcal{D}\tfrac{\delta H_{\text{KdV}}}{\delta v} \), where \( \mathcal{D} = -\partial_x \) and \( \mathcal{H}_{\text{KdV}} = \int_{\mathbb{R}} ( v^3 - \tfrac{\gamma}{2} v_x^2 ) dx \). 
Just as canonical Hamiltonian systems have a phase space variational description, this noncanonical system can be deduced from the following Lagrangian
\(
	\ell_\text{KdV}(v) = \int_{\mathbb{R}} \tfrac12 v\partial_t(\mathcal{D}^{-1}v) - \mathcal{H}_\text{KdV}(v) \,dx =  \int_{\mathbb{R}} \tfrac12 v\partial_t(\mathcal{D}^{-1}v) - v^3 + \frac{\gamma}{2} v_x^2 \,dx \,.
\)
}
.
To apply Hamilton's Principle to this action, we vary the KdV variable $v$ arbitrarily and enforce the variation of the Burgers' velocity $u$ to be the Euler-Poincar\'e constrained variation $\delta u = \p_t w - {\rm ad}_uw$, with ${\rm ad}_uw=-[u,w]:=-uw_x+w u_x$. With these variations, one finds
\begin{align*}
	0 &= \delta S_\text{WCI} 
    =  \int_{t_0}^{t_1} \int_{\mathbb{R}} (\p_t  + \mathcal{L}_u)\left( u {- \frac12 v^2} \right)\, w + \left( \mathcal{D}^{-1}v_t {- uv} - 3v^2 - \gamma v_{xx} \right)\,\delta v \,dx\,dt \,,
\end{align*}
where the Lie derivative $\mathcal{L}_u$ acting on 1-form densities, which emerges by integrating ${\rm ad}_u$ by parts in $x$ to obtain ${\rm ad}^*_u$, is defined in coordinates by $\mathcal{L}_u m = {\rm ad}^*_um = m\p_x u+\p_x(mu)$ where $m=u - \tfrac12 v^2$. Since $w$ and $\delta v$ are arbitrary, the following coupled system emerges:
\begin{equation}\label{Burgers–swept KdV-EP+EL}
    {\partial_t}m + (m {\partial _x} 
    + {\partial _x}m)u = 0 \,, \quad
    {\partial_t}v 
    + {\partial _x}({uv +} {\gamma{v_{xx}} + 3{v^2}} ) = 0\,.
\end{equation}
The Hamiltonian form of the Burgers–swept KdV equations, when written in terms of the velocity variables $(u,v)$ may be expressed as a set of semidirect-product Lie-Poisson equations, extended by a generalised 2-cocycle. Namely, 
\begin{align}
 \partial_t
 \begin{bmatrix}
      u\\
      v   
 \end{bmatrix}
 = -
 \begin{bmatrix}
     \mathcal{L}_\square u & v\partial_x\square\,   \\
     \partial_x (v\square)  & \partial_x\square 
 \end{bmatrix}
 \begin{bmatrix}
      \delta{h}/\delta{u} = u\\
      \delta{h}/\delta{v} =  \gamma v_{xx} + 3 v^2
 \end{bmatrix}
\,,\label{eqn-with4}
\end{align}
where the $\square$ denotes the location of the matrix multiplication,
with Hamiltonian given by the sum of the Burgers' and KdV Hamiltonians as
\begin{align}
h(u,v)=\int \frac12 u^2 + v^3- \frac{\gamma}{2} v_x^2  \, dx
\,.\label{eqn-TotHam}
\end{align}
\begin{remark}
The Burgers–swept KdV system \eqref{eq: Burgers–swept KdV} can also be derived by transforming the KdV Lagrangian into the frame of motion of the Burgers dynamics before taking variations: 
\begin{align}
0 = \delta S = \delta \!\! \int_{t_0}^{t_1} \int_{\mathbb{R}}  \ell_{B}(u)+\ell_{K d V}(n, \nu)\,dt
=\frac12\delta\!\! \int_{t_0}^{t_1} \int_{\mathbb{R}}  u^2 
+  \nu n_x+2 n_x^3 - \gamma n_{x x}^2 \,d x \,dt \,,
\label{eqn-SumAction}
\end{align}
with Eulerian fluid variable $u=g_t g^{-1}$, KdV velocity potential $n=\phi g^{-1}$ and KdV velocity in the Burgers frame $\nu=\phi_t g^{-1}$ with constant $\gamma$. Using the variations $w  = \delta g{g^{ - 1}}$, $\xi  = \delta \phi {g^{ - 1}}$ and the induced constrained variations $\delta u = {w _t} + w u - uw ; \delta n = \xi  - nw ;\delta \nu  = {\xi _t} + \xi u - \nu w$. Thus, equation \eqref{eqn-SumAction} places the Burgers–swept KdV system into an Euler-Poincar\'e systems, as discussed in section 2 of \cite{HOLM2023133847}.
\end{remark}
\begin{remark}
For more discussion of the variational derivation of the Burgers–swept KdV system, see Appendix \ref{app-derivingBurgers–swept KdV}.
\end{remark}

\subsection{Distribution of Momentum}
The momentum in the Burgers–swept KdV system is shared between its Burgers and KdV components. The total momentum 1-form density \(m\,dx\otimes dx\) of the Lagrangian is obtained by varying the total Lagrangian with respect to vector field \(u\partial_x\), yielding  
\begin{align}
m dx\otimes dx  := \frac{\delta \ell}{\delta u} =   udx\otimes dx - \frac{1}{2} v^2 dx\otimes dx,
\label{Total-BKdV-mom}
\end{align}
which is the difference between the Burgers and KdV momentum 1-form densities. We can rewrite the Hamiltonian with a diagonalised Poisson matrix with ${m}=u-\frac12 v^2$.
\begin{remark}
The Burgers–swept KdV system can be equivalently expressed in terms of ${m} = u- \frac12 v^2$ with Hamiltonian
\begin{equation}
    h({m} ,v) = \int_{\mathbb{R}} {\frac{1}{2} ({m}+\frac12 v^2)^2 + {v^3} - \frac{\gamma}{2}v_x^2 \,dx} \,.
\end{equation}
Consequently, the Burgers–swept KdV system \eqref{eq: Burgers–swept KdV} or, equivalently, \eqref{v+ep-eqn-intro} may also be written as
\begin{equation}
\partial_t \begin{bmatrix}
      m\\
      v   
 \end{bmatrix}
 =-\left[\begin{array}{cc}
\mathcal{L}_\square m & 0 \\
0 & \partial_x
\end{array}\right]
\left[\begin{aligned}
&\delta h / \delta {m}={m} + v^2 / 2 = u \\
&\delta h / \delta v=\gamma v_{x x}+3 v^2+v^3 / 2+v {m}
\end{aligned}\right].
\label{LP-brkt}
\end{equation}
\end{remark}

The $v$-equation here is Gardner's equation \citep{Gardner1967} with an extra coupling term $(m v)_x$. The $m$-equation is expressed in the geometric form, in terms of the coadjoint representation of vector fields on their dual, as a Lie-transport equation with $u=\frac12 v^2 + m$:
\begin{equation}
    {m}_t + (m\partial_x + \partial_x m) {u} 
    = {m}_t + \mathcal{L}_{u} {m}
    =0
\,.\label{v+ep-eqn2}
\end{equation}
The momentum density $m$ is advected by the flow of $u$, carried by the Burgers velocity field $u \partial_x$ as a 1-form density $m d x^2$. This transport is expressed as the Lie derivative $\mathcal{L}_u\left(m d x^2\right)$.

\begin{proposition}
\label{prop: m=0}
    In the Burgers–swept KdV system \eqref{eq: Burgers–swept KdV}, if initially ${m}(x,0)=u(x,0)-\tfrac12 v^2(x,0)=0$, then ${m}(x,t)=0$ for all time. Thereafter, the Burgers–swept KdV system reduces to Gardner's equation.
\end{proposition}

\subsection{Compound Soliton and Travelling Wave Solutions}
\noindent Travelling waves in the Burgers–swept KdV system reflect the balance between nonlinear advection and dispersion in $u$ and $v$. Using the Galilean-invariant coordinate $\xi=x-c t$, we set $u(t, x)=u(\xi), v(t, x)=v(\xi)$. Substituting this Ansatz into \eqref{eq: Burgers–swept KdV} reduces the system to
\begin{equation}
 - cu' + 3uu' =  - v(\gamma v'' + 3{v^2})' \,,\quad\hbox{and}\quad
 - cv' + (uv)' =  - (\gamma v'' + 3{v^2})' \,.
\label{eq: Burgers–swept KdV travelling wave}
\end{equation}
The travelling wave solution has the following relation between $u$ and $v$.
\begin{lemma}[Burgers–swept KdV travelling wave condition]
\label{lemma:1}
Let $u(\xi)$ and $v(\xi)$ be travelling wave solutions for \eqref{eq: Burgers–swept KdV travelling wave}, with boundary conditions $\mathop {\lim }\limits_{\xi  \to  \pm \infty } u(\xi) = {u_0}$ and $\mathop {\lim }\limits_{\xi  \to  \pm \infty } v(\xi) = 0$, then:
   \begin{equation}
   \label{eq: travelling_wave_condition}
    {(u - c)^2}(u - \tfrac{1}{2}{v^2}) = (u_0-c)^2u_0 \,.
\end{equation}
\end{lemma}
For solitons, where the solution and its derivatives vanish as $\xi \rightarrow \pm \infty$, we have $u_0=0$. This implies $u= \frac{1}{2} v^2$ and $m=0$, reducing \eqref{Burgers–swept KdV-EP+EL}
to an integrable form that yields the following explicit compound soliton solution.
\begin{theorem}[Compound Soliton Solution]
\label{thm:1}
The Burgers-swept KdV system has the following compound soliton solution in a closed form, for constant speed $c$ and $\xi = \frac{(x - ct)}{\sqrt{\gamma}}$:
\begin{equation}
\begin{split}
v(\xi ) &= \frac{{c}}{{{w(\xi)^2} + 1}}\,{{{\mathop{\rm sech}\nolimits} }^2}( \pm \frac{\sqrt c}{2} \xi  + \varphi )\,,\\
u(\xi)&=\frac{{c^2}}{{2({w(\xi)^2} + 1)^2}}{{{\mathop{\rm sech}\nolimits} }^4}( \pm \frac{\sqrt c}{2} \xi  + \varphi ) \,,
\end{split}
\label{Gardner-redux}
\end{equation}
where $w(\xi ) = \frac{{\sqrt c }}{2}\tanh \left( { \pm \frac{{\sqrt c }}{2}\xi  + \varphi } \right) + \frac{{\sqrt {c + 4} }}{2}$ and $\varphi$ is the phase that specifically satisfies: \\
1. Locality:  $u, u', v, v' \rightarrow 0$ as $\xi \rightarrow \pm \infty$,\\
2. Stability: the shapes of $u$ and $v$ persist,\\
3. Interaction Properties: Interaction between compound solitons 
preserves velocities, producing only a phase shift in the wave trajectories.
\end{theorem}

\begin{proof}
By Proposition \ref{prop: m=0}, the solution satisfies the Burgers–swept KdV equations under the given boundary conditions, with ${m}=u-\frac{1}{2} v^2=0$ invariant during evolution by the advection equation \eqref{v+ep-eqn2}. For $m=0$, the Burgers–swept KdV system reduces to the Gardner equation, whose soliton solution corresponds to the $v(\xi)$ profile satisfying the properties of a soliton solution.
\end{proof}
The compound soliton in \eqref{Gardner-redux} shares a similar profile with the classic KdV soliton, \(v_\text{KdV}(\xi) = \frac{c}{2} \operatorname{sech}^2(\pm \frac{\sqrt{c}}{2} \xi + \varphi)\) but modulated by an interaction factor $w$. The factor $w \geq 1$ reduces the amplitude of $v$ in compound solitons relative to KdV solitons at the same speed $c$, with positive $u$ aiding the propagation. The dispersive coefficient $\gamma$ scales the soliton dynamics in space and time. When $u_0 > 0$, we have the following result of a compound solitary wave solution.

\begin{remark}
When $c=1$, the interaction factor $w(\xi)$ ranges within $\left(\frac{1}{\varphi}, \varphi\right)$, where $\varphi=\frac{1+\sqrt{5}}{2} \approx 1.618$ denotes the Golden Ratio.
\end{remark}

\begin{remark}[Compound Solitary Wave Solution with $u_0>0$]
The Burgers-swept KdV system admits a compound solitary wave solution $u(\xi), v(\xi)$ propagating with constant speed $c$ with asymptotic condition $\underset{\xi \rightarrow \pm \infty}{\lim} u(\xi)=u_0>0, \quad \underset{\xi \rightarrow \pm \infty}{\lim} v(\xi)=0$. The solution satisfies the following implicit relation given by an ODE: 
\begin{equation}
    (u(v)+3v-c) v + v_{xx} =\text{const} \,,  
    \label{eq: solitary_wave}
\end{equation}
    where $\xi = x-ct$ and $u(v)$ is a fixed solution of the cubic equation \eqref{eq: travelling_wave_condition} in Lemma \ref{lemma:1}.
\end{remark}
The wave profile can be computed numerically, revealing its structure (Figure \ref{fig:soliton_comparison}). 
\begin{figure}[H]%
    \centering
    \begin{subfigure}[b]{0.48\textwidth}
        \includegraphics[width=\linewidth]{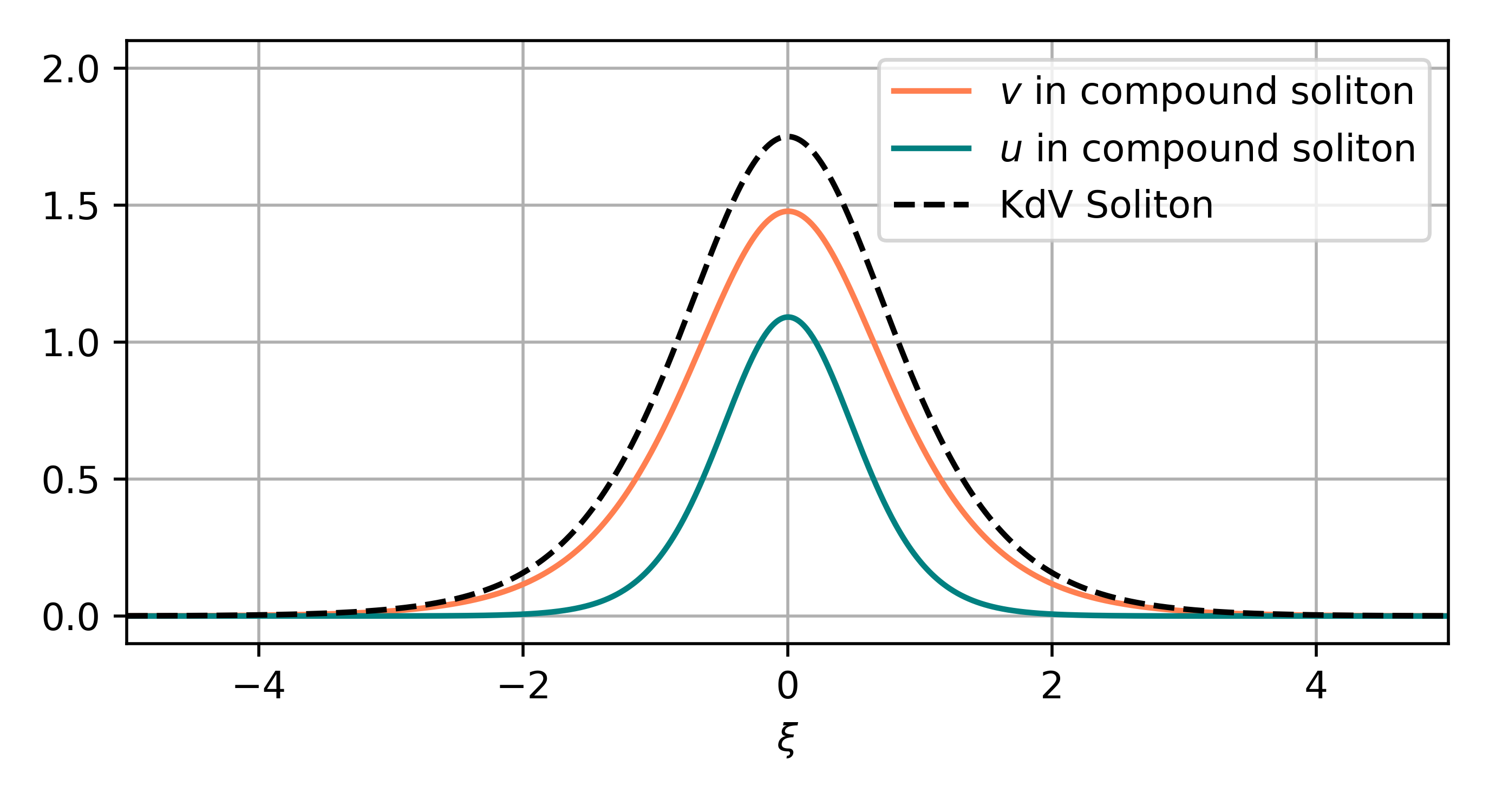}
    \end{subfigure}
    \hfill
    \begin{subfigure}[b]{0.48\textwidth}
        \includegraphics[width=\linewidth]{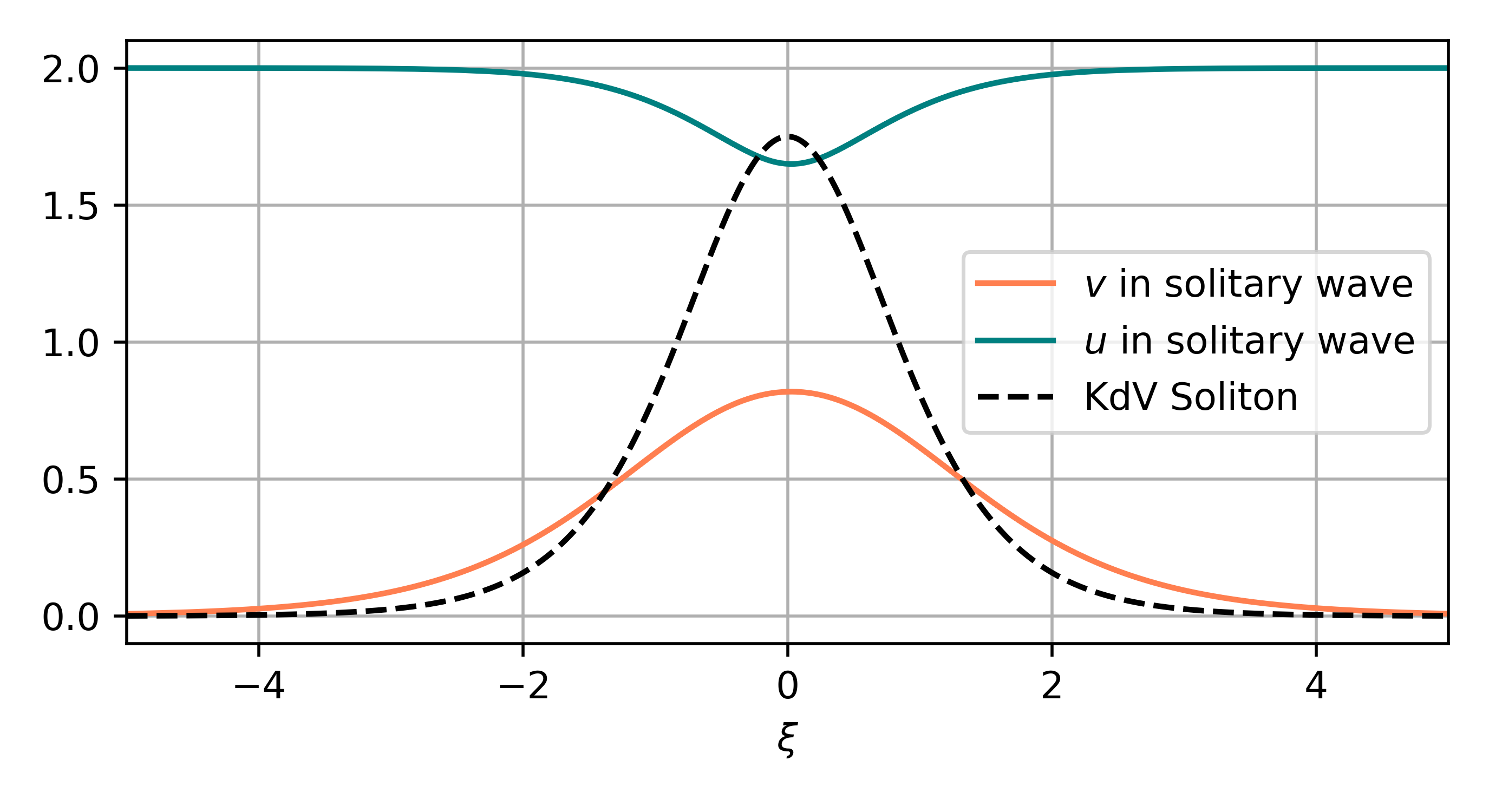}
    \end{subfigure}
    \caption{Comparison between the Burgers–swept KdV compound soliton (left) / solitary wave (right) and the KdV soliton $v_\text{KdV}(\xi)$ at the same wave speed $c=3.5$, under two boundary conditions $\lim _{x \rightarrow \pm \infty} u(x)=u_0$: $u_0=0$ (left) for \ref{Gardner-redux} and $u_0=2$ (right) for \ref{eq: solitary_wave}. 
}
    \label{fig:soliton_comparison}
\end{figure}
The model also has the following three solutions.
\begin{figure}[H]
    \centering
    \includegraphics[width=1\linewidth]{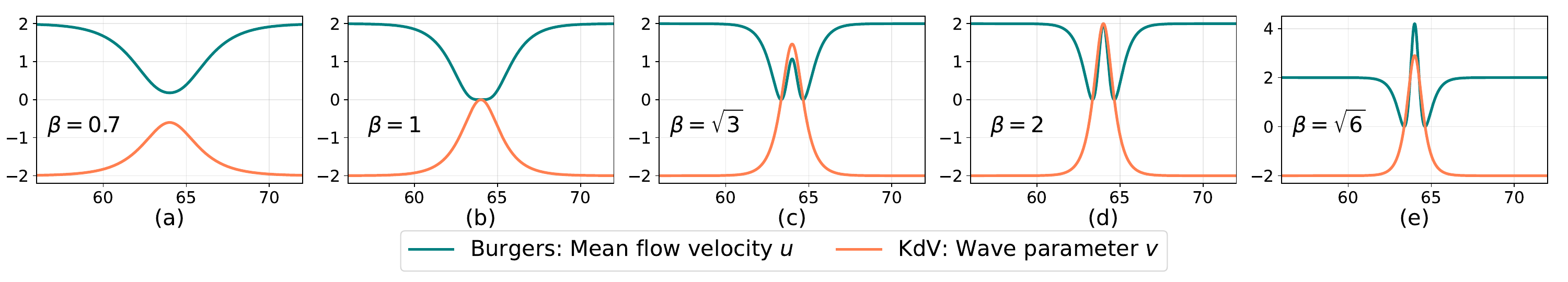}
    \caption{$\operatorname{dn}[\beta \xi, m]$ travelling wave solution profiles in the Burgers–swept KdV system with elliptic parameter  $m=1$. The Burgers component $u$ (mean flow velocity, blue) and KdV component $v$ (wave parameter, orange) are plotted for increasing values of the wave parameter $\beta$ (a) $\beta=0.7$, (b) $\beta=1$, (c) $\beta=\sqrt3$, when $\beta>1$ a W-shape structure appears in $u$  (d) $\beta=2$, (e) $\beta=\sqrt6$ is corresponding to the case $c=0$ (\textit{no propagation}), the wave profile appears stationary, with a highly localized, peaked structure. This profile resembles the \textit{Peregrine soliton}, a localized wave in the nonlinear Schrödinger equation.}
    \label{fig: alternative_soliton}
\end{figure}
\begin{theorem}[Periodic wave]\label{Thm2.2}
    Burgers-swept KdV system \eqref{eq: Burgers–swept KdV travelling wave} admits the following periodic solutions with relation $u=v^2/2$, \[
\begin{aligned}
    v_1(\xi) &= -2 + 4\beta^2 \operatorname{dn}\left[\beta \xi, m\right], \\
    v_2(\xi) &= -2 + 4\beta^2 \sqrt{m} \operatorname{cn}\left[\beta \xi, m\right], \\
    v_3(\xi) &= -2 + 2\beta^2 \operatorname{dn}\left[\beta \xi, m\right] \pm 2\beta^2 \sqrt{m} \operatorname{cn}\left[\beta \xi, m\right].
\end{aligned}
\]
In these periodic solutions $\xi = x-ct+\varphi$, with speed $c$; arbitrary phase $\varphi$; elliptic Jacobi function parameter $m$; and amplitude factor $\beta^2 = {(c + 6})/({2 - m)}$. 
\end{theorem}
Figure \ref{fig: alternative_soliton} illustrates the periodic wave solutions, with $\beta=\sqrt6$ serving as a special case where the wave does not propagate $(c=0)$ but remains stable and stationary.  

\section{Simulations of the Burgers–swept KdV system}\label{sec-3}
The simulations below are performed using the pseudo-spectral method supported within the ``Dedalus Project" \citep{2020PhRvR...2b3068B}. The Burgers–swept KdV system contains a Burgers-type equation with a shock wave type solution, and a viscosity term $\eta u_{xx}$ with $\eta =0.02$ is added to the Burgers part of the equation \eqref{eq: Burgers–swept KdV} as $u_t + uu_x = \eta u_{xx} -v{\partial _x}(  3{v^2} + {\gamma{v_{xx}}} )$ to form a bore front.
\subsection{The emergent behavior of Burgers–swept KdV Compound soliton }
Burgers–swept KdV solitary waves extend KdV type solitons by coupling dynamics in $u$ and $v$. Even with nonzero initial data in $v$ alone ($u=0$), the system evolves to form Burgers–swept KdV compound solitons as shown in Figure~\ref{fig: kdv_emerge}. That is for $m \neq 0$ initially, the system tends to self-organize into localized regions where $m \rightarrow 0$ and restores Gardner-type dynamics and integrability, allowing compound solitons to emerge spontaneously. The train of compound solitons suggests a fruitful future study by using the inverse scattering transform on this type of wave-current interaction model.

\begin{figure}[H]
    \centering
    \includegraphics[width=\linewidth]{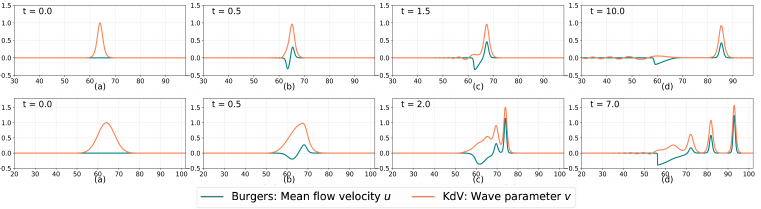}
    \caption{Formation of Burgers–swept KdV compound solitons from $m\ne0$ initializations. 
    (a) Initial condition with $u=0$ and a KdV soliton $v= v_{\mathrm{KdV}}$ with $c=2$ (top), Gaussian packet $v=\exp (-x^2 / 32)$ (bottom). (b) The Burgers component develops an N-wave structure.
    (c) The right-propagating peak in $u$ aligns with the part in $v$, initiating compound structure formation. (d) A single (top)/ train (bottom) of compound solitons emerges and travels to the right, leaving residual waves to the left. For spacetime plot see appendix \ref{appendix: spacetime}.
}
    \label{fig: kdv_emerge}
\end{figure}

\subsection{Interaction between a Burgers Bore and a Compound Soliton}

\textbf{Refraction. }
Figure \ref{fig: refraction} illustrates a refraction-type interaction in the Burgers–swept KdV system, where a compound solitary wave crosses the interface between two distinct wave propagation environments: outside the Burgers current, the stable structure is a compound soliton, while inside the current, it takes the form of a solitary wave. As the wave exits or enters the Burgers current (i.e., crosses the bore front), it transitions between these two solution types. This transition leaves behind a train of residual $v$-waves in the current region, which, in the frame of the bore, resembles a Gardner-type wave train \citep{KI2012}. 
\begin{figure}[H]
    \centering
    \includegraphics[width= \linewidth]{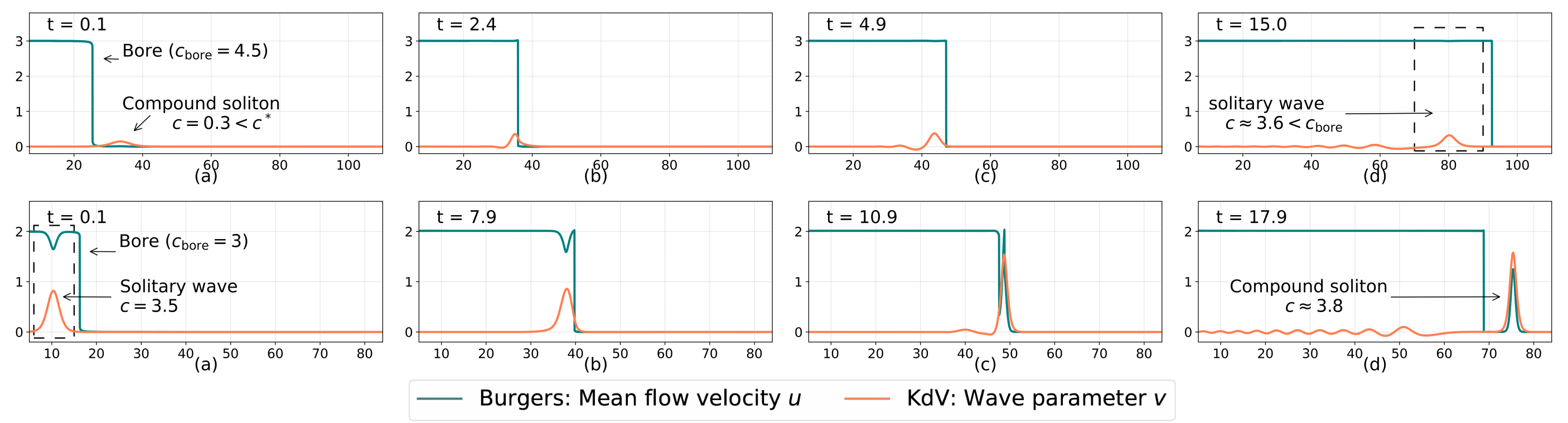}
    \caption{Refraction type interaction between bore and a compound wave. For spacetime plot see appendix \ref{appendix: spacetime}.\\
    \textbf{Top}: a compound soliton satisfying \ref{Gardner-redux} swept by the Burgers current and forms a solitary wave slower than the bore. (a) a compound soliton ($c=0.3$) out of the Burgers current ($u_0=3$). (b) - (c) the bore takeovers the soliton (d) a solitary wave remains in the current. \\
    \textbf{Bottom}: a compound solitary wave satisfying \ref{eq: solitary_wave} in the Burgers current exits current and forms a faster compound soliton. (a) a compound solitary wave  with speed $c=3.5$ within a Burgers current ($u_0=2$) (bore front speed $c=1.5u_0$). The solitary wave approach the bore front at speed $c=0.5$ without much shape deformation. (b) the solitary wave arrives the bore front and starts to detach the bore, (c) the solitary wave carries away a piece of the bore and start to leave the bore front, (d) a compound soliton faster than the bore is formed with a `train' of waves carried along in the `current' behind the bore front.}
    \label{fig: refraction}
\end{figure}

The key feature of this refraction phenomenon lies in the change of propagation speed relative to the bore front (represented by the vertical line in figure \ref{fig:phase1} (left)). When viewed in the moving frame of the front, the soliton’s velocity changes across the interface, analogous to the classical refraction. Faster solitary waves inside the current may eventually escape and emerge ahead of the bore as compound solitons, while slower compound solitons outside the current are refracted and fall behind the bore under a threshold discussed in the next section. 

\textbf{Reflection. }
When the compound soliton speed is above a certain threshold, the main profile will not be `refracted' by the bore but will be `reflected' and carry away part of the bore to create a compound soliton which leaves the bore. 
\begin{figure}[H]
    \centering
    \includegraphics[width=\linewidth]{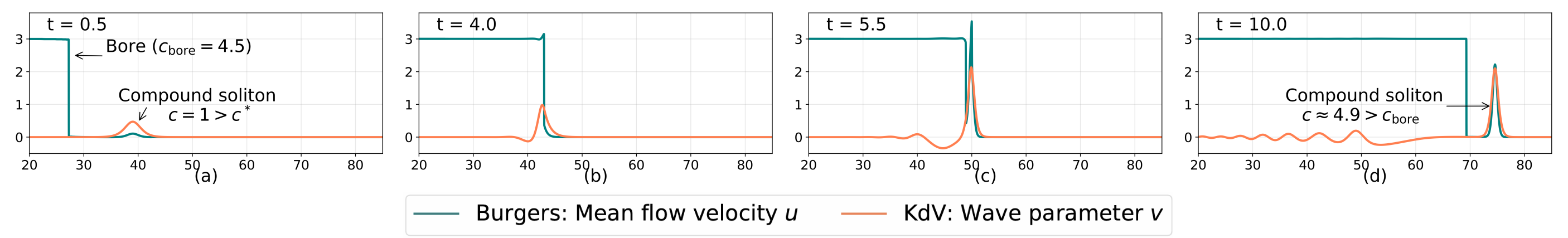}
    \caption{Reflection type interaction that a compound soliton outside the Burgers bore current get swept and form a faster compound soliton leaving the bore. (a) Burgers bore ($u_0=3$) approaches to a compound soliton ($c=1$).  (b) - (c) The compound soliton interacts with the bore nonlinearly and gains amplitude $u, v$. (d) A faster compound soliton is formed which leaves the bore. For spacetime plot see appendix \ref{appendix: spacetime}.}
    \label{fig:Reflection}
\end{figure}
\begin{figure}[H]
    \centering
    \includegraphics[width=.9\linewidth]{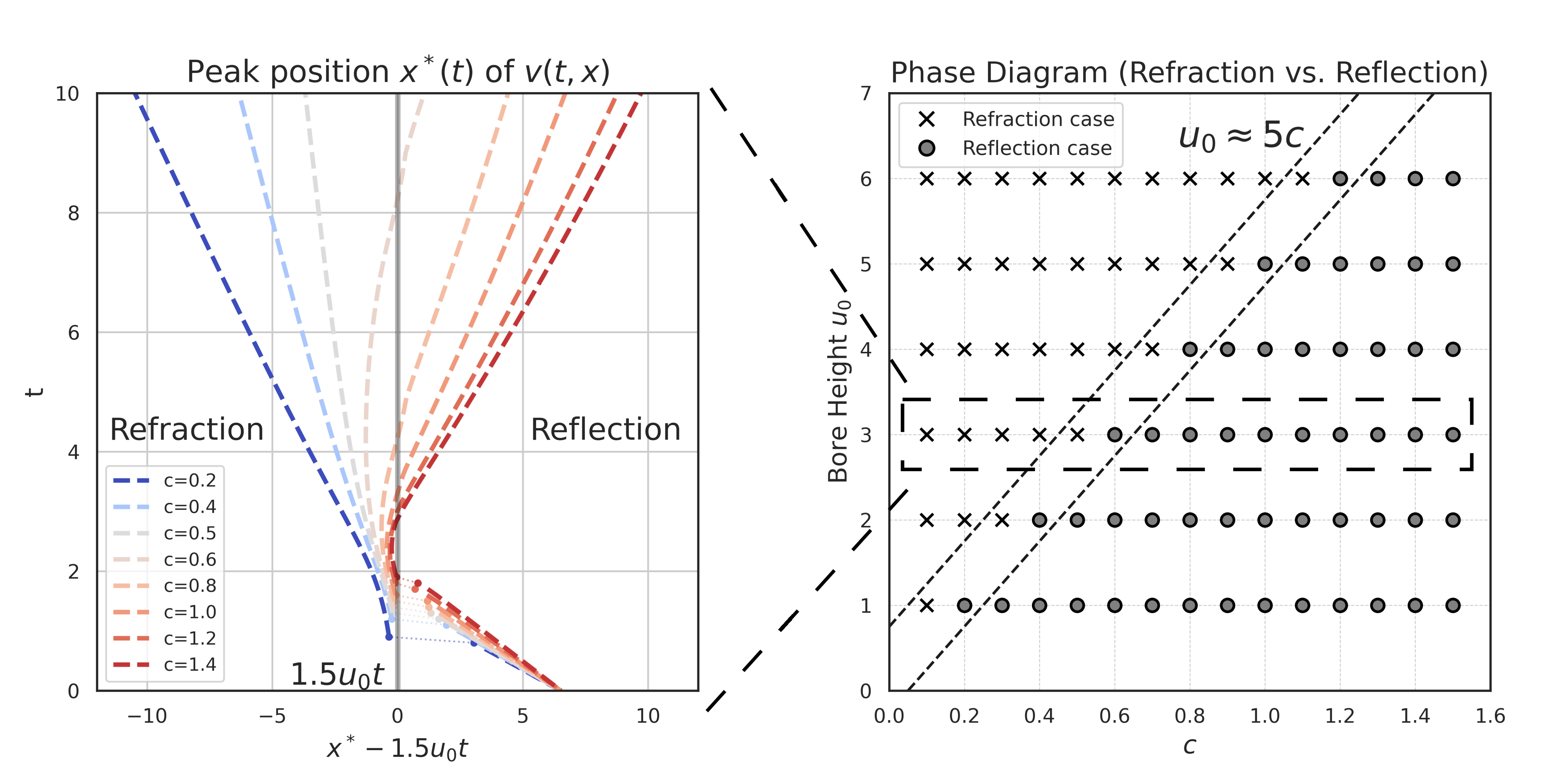}
    \caption{\textbf{Left}: A series of simulations showing a Burgers bore ($u_0=3$) overtaking compound solitons with different speed in the bore's moving frame. Dashed lines indicate the  position $x^*(t)=\operatorname{argmax}_x v(t)$ of the soliton/solitary wave peak over time.\textbf{Right}: Summary of refraction and reflection type outcomes for different $u_0$ and $c$ values under the viscosity $\eta=0.01$. The dotted line indicates the approximate threshold $c^*\approx u_0/5$ separating the two regimes.}
    \label{fig:phase1}
\end{figure}
When the speed of an incoming compound soliton exceeds a critical threshold, it is no longer refracted and absorbed into the bore. Instead, the soliton gets `reflected' and forms a faster compound soliton carrying away part of the bore. The critical threshold for the soliton speed $c$, as shown in Figure~\ref{fig:phase1}, follows an approximately linear relation $u_0 \approx 5c$ for small values of $u_0$ and $c$. For a given soliton, a slower overtaking bore leads to a longer interaction time near the front, enabling the soliton to reorganize and finally escape.

\textbf{Soliton fusion. }
\begin{figure}[H]%
    \centering
    \includegraphics[width=\linewidth]{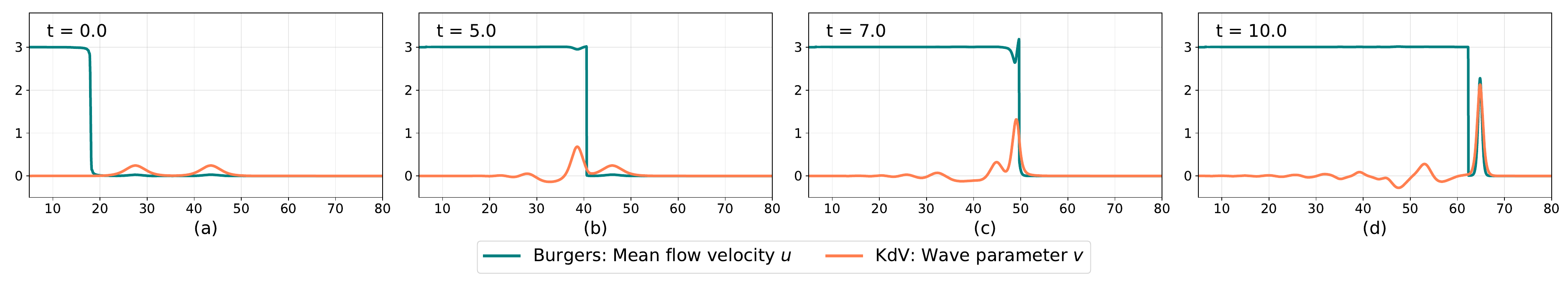}
    \caption{Fusion of two weak compound solitons ($c = 0.5$) is triggered by the overtaking of a Burgers bore ($u_0 = 3$). (a) The bore approaches two refracting compound solitons. (b) The bore current captures the left soliton near the bore front. (c) Before the left soliton falls behind the bore front, it is carried into the collision with the right soliton. (d) After the soliton fusion, one faster compound soliton forms and propagates ahead. For spacetime plot see appendix \ref{appendix: spacetime}.
}
    \label{fig:fusion}
\end{figure}

\begin{figure}[H]%
    \centering
    \includegraphics[width=\linewidth]{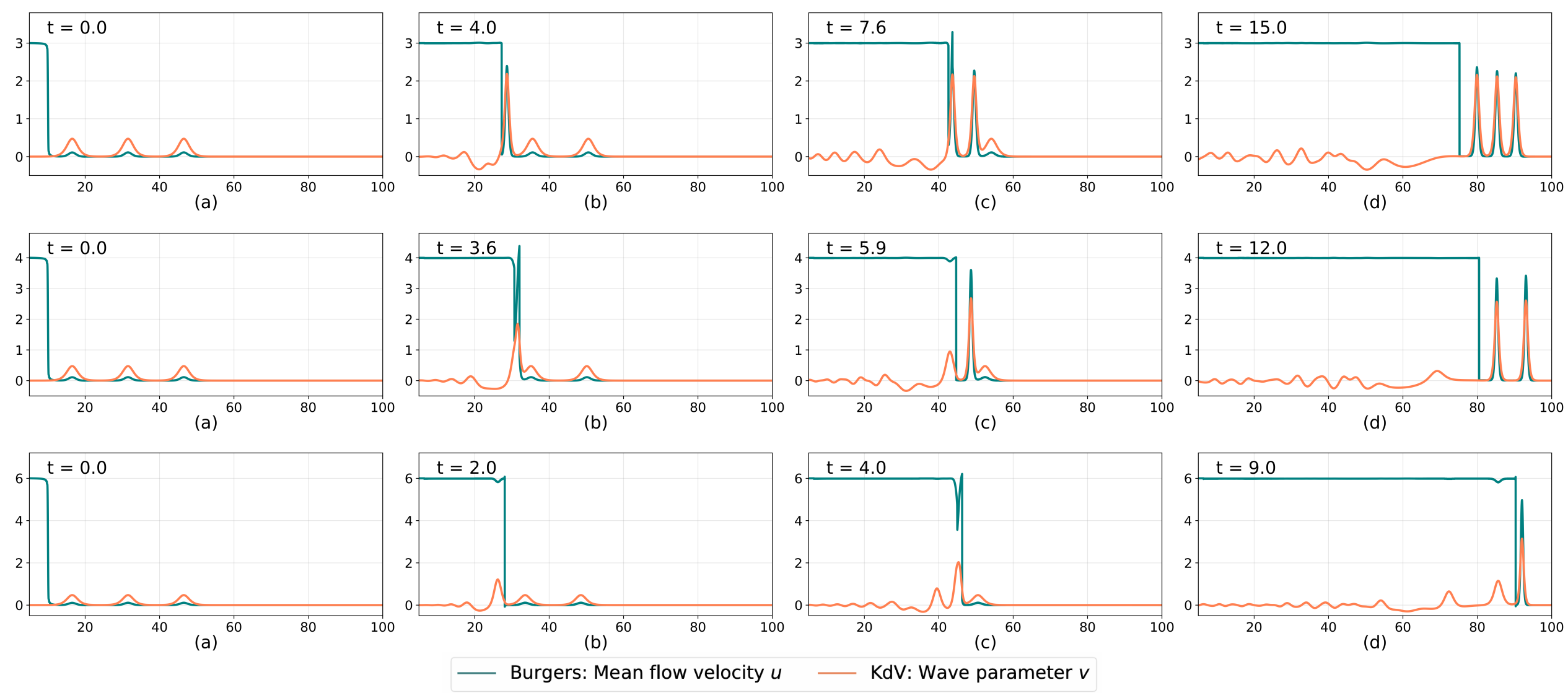}
    \caption{Three trials of bores with heights $u_0=3,4,6$ overtaking the same initial set of three slower compound solitons; each with speed $c=1$ and separated at a distance of 15. One sees that the different initial bore heights in the overtaking interactions result in the creation of different numbers of compound solitons leaving the bore. These results are classified as either reflections or refractions, as observed in the frame of the moving bore. The $u_0=3$ case (Top) yields three compound solitons reflected. The $u_0=4$ case (Middle) yields two compound solitons reflected and one refracted. The $u_0=6$ case (Bottom) yields one compound soliton reflected and two refracted. Each row shows the bore before interacting with (a)  the first soliton, (b) the second soliton, (c) the third soliton, and (d) after all interactions.
}
    \label{fig: multi_soliton_fusion}
\end{figure}
Near the refraction–reflection threshold, the interaction time between a compound soliton and the bore becomes significantly extended. In this regime, a soliton can remain near the bore front for a long time without immediately escaping or falling behind. If it encounters another incoming soliton during this prolonged interaction, the two may fuse into a single, stronger soliton with enough momentum to escape the bore.  As shown in Figure~\ref{fig:fusion}, the bore ($u_0=3$) overtakes two weak compound solitons ($c=0.5$)—each of which would be refracted individually (see Figure~\ref{fig:phase1}), but releases one single, faster, merged soliton, representing a reflection-type interaction with a different number of outgoing solitons. This shows the highly nonlinear nature of the Burgers–swept KdV system.
Figure~\ref{fig: multi_soliton_fusion} highlights the nontrivial outcomes of overtaking interactions between a Burgers bore and a sequence of compound solitons. Even when the incoming solitons are identical and evenly spaced, the number of outgoing solitons observed outside the bore can vary, depending on the bore strength $u_0$ and interaction timing. In this sense, the bore acts as a soliton fusion mechanism, combining multiple incoming solitons into fewer outgoing ones. This reduction in soliton number reflects the nonlinear and complex nature of wave–current interactions in the Burgers–swept KdV system. 

\section{Summary, Conclusion and Open Problems}
In this work we have considered the derivation and numerical simulation of the Burgers–swept KdV system shown in \eqref{eq: Burgers–swept KdV}. The main contribution of this work are summarised as follows. 

In Section \ref{sec-2}, we considered a Doppler shifted variational principles to derive the Burgers–swept KdV system in \eqref{eq: Burgers–swept KdV}. The Hamiltonian structure of the Burgers–swept KdV system in equation \eqref{LP-brkt} reveals the Lie transport nature of the Burgers–swept KdV system. Namely, the quantity $m=u-\tfrac12 v^2$ is transported by the bore velocity $u$ and the KdV velocity $v$ satisfies a conservation equation which for $m=0$ reduces to Gardner's equation. Thus, we recognised the Burgers–swept KdV system as an extension of the celebrated Gardner equation by transforming variables to the equivalent Burgers–swept KdV system in equation \eqref{v+ep-eqn-intro}. Additionally, we derived the compound soliton solution (Theorem \ref{thm:1}) and the travelling wave soliton solutions of the Burgers–swept KdV system (Theorem \ref{Thm2.2}).  

In Section \ref{sec-3}, numerical simulations reveal a range of nonlinear wave-current interaction (WCI) phenomena, including the emergence of compound Gardner solitons from nonlinear interactions between Burgers bores and compound waves. 
These numerical simulations show Burgers–swept KdV dynamics of Burgers bore and KdV wave velocity profiles undergoing nonlinear interactions which develop into localised compound structures  classified broadly as refraction or reflection. Under reflection, these localized structures separate and propagate rightward, leaving the others behind. Asymptotically, these separated, rightward-moving localized structures become rightward-moving trains of compound Gardner solitons. 

Future work may investigate how these recurrences of Gardner-equation behaviour occur as the system transitions locally between integrable and non-integrable regimes. This is a question of both mathematical interest as an infinite-dimensional version of the FPUT problem, and of physical relevance in the broader context of nonlinear wave dynamics in fluid mechanics and plasma physics where Gardner's equation appears.

\subsection*{Acknowledgements} 
We are grateful to Collin Cotter, Dan Crisan, Theo Diamantakis, Boris Khesin, David Levermore and James Woodfield for several thoughtful suggestions. DH and RH were partially supported by Office of Naval Research grant N00014-22-1-2082. DH, OS, and HW were also supported by European Research Council Synergy Grant DLV-856408. OS acknowledges funding for a research fellowship from Quadrature Climate Foundation. 

\bibliographystyle{plainnat}
\bibliography{bibliography}%

\newpage

\appendix
\section{Variational derivations of Burgers, KdV and the Burgers–swept KdV system}
\label{app-derivingBurgers–swept KdV}

\subsection{Variational derivations of Burgers and KdV equations}
This appendix begins by briefly mentioning the geometric structures of the Burgers' and KdV equations. It then provides two auxialary derivations of the 
\smallskip

\paragraph{Burgers' equation in one dimension.}The Burgers' equation follows from applying Hamilton's Principle to the action corresponding to the kinetic energy Lagrangian
\begin{equation}
	\ell_{B}(u) = \int_{\mathbb{R}} \frac{u^2}{2}\,dx \,,
\end{equation}
where the variations in $u$ are taken to be \emph{Lin constrained} Euler-Poincar\'e variations given by
\begin{align*}
    \delta u = {\eta _t} + \eta u - u\eta  = \partial_t \eta - {\rm ad}_u \eta \,,
\end{align*}
for an arbitrary vector field $\eta \in \mathfrak{X}(\mathbb{R})$. In the relation above, ${\rm ad}$ denotes the standard adjoint action of the space of vector fields in itself, inherited from the group structure of the diffeomorphism group.  This variation follows from taking arbitrary variations, vanishing at the endpoints, of the flow $g_t \in {\rm Diff}(\mathbb{R})$ related to the velocity by $u = \dot{g}g^{-1}$, where concatenation has been used to denote tangent lifted group translation. The arbitrary vector field in the variation is then given by $\eta = \delta g g^{-1}$. The equation follows from Hamilton's Principle as
\begin{align*}
	0 &= \delta\int_{t_0}^{t_1} \ell_B(u)\,dt = \int_{t_0}^{t_1}\int_{\mathbb{R}} u \,\delta u \,dx\,dt
	\\
	&=  \int_{t_0}^{t_1}\int_{\mathbb{R}}  u (\partial_t \eta - {\rm ad}_u) \,dx\,dt =  \int_{t_0}^{t_1}\int_{\mathbb{R}} (\p_t u + {\rm ad}^*_uu)\,\eta\,dx\,dt \,,
\end{align*}
where ${\rm ad^*}$ is the dual operator to ${\rm ad}$ with respect to the duality pairing. Since $\eta$ is arbitrary, the fundamental lemma of the calculus of variations gives the equation
\begin{equation}
	0 = \p_tu + {\rm ad}^*_uu = \p_t u + u\p_xu + \p_x(u^2) = u_t + 3uu_x \,.
\end{equation}
\smallskip

\paragraph{The KdV equation.} The KdV equation has a noncanonical Hamiltonian structure. That is, it can be expressed in the form
\begin{equation}
	\p_t v = \mathcal{D}\frac{\delta H_{KdV}}{\delta v} \,,\quad\hbox{where}\quad \mathcal{D} = -\p_x \,,\quad\hbox{and}\quad H_{KdV} = \int_{\mathbb{R}} v^3 -  \frac{\gamma}{2} v_x^2\,dx \,.
\end{equation}
Indeed, computing the variational derivative of the Hamiltonian, this implies
\begin{equation}
	v_t = -\p_x\left( 3v^2 + \gamma v_{xx} \right) = -6vv_{x} - \gamma v_{xxx} \,,
\end{equation}
which is the well known KdV equation. Just as canonical Hamiltonian systems have a phase space variational description, this noncanonical system can be deduced from the following Lagrangian
\begin{equation}
	\ell_{KdV}(v) = \int_{\mathbb{R}} \frac12 v(\mathcal{D}^{-1}v)_t - H_{KdV}(v) \,dx =  \int_{\mathbb{R}} \frac12 v\mathcal{D}^{-1}v_t - v^3 + \frac{\gamma}{2} v_x^2 \,dx \,.
\end{equation}
Indeed, applying Hamilton's Principle, we find that
\begin{align*}
	0 &= \delta\int_{t_0}^{t_1} \ell_{KdV}(v)\,dt =  \int_{t_0}^{t_1} \frac12 \delta v\mathcal{D}^{-1}v_t + \frac12 v \mathcal{D}^{-1}\delta v_t - 3v^2\delta v + \gamma v_x\delta v_x \,dx\,dt
	\\
	&= \int_{t_0}^{t_1}\int_{\mathbb{R}} \delta v\left(\mathcal{D}^{-1}v_t - 3v^2 - \gamma v_{xx} \right) \,dx \,,
\end{align*}
and since $\delta v$ is arbitrary this implies the KdV equation.

\subsection{Variational derivation of the Burgers–swept KdV system}
\label{CoM-app}

This variational derivation of the Burgers–swept KdV system in \eqref{eq: Burgers–swept KdV} proceeds by introducing the composition of two time-dependent maps $g$ and $\phi$ in defining the following Lagrangian:
\begin{align}
\begin{split}
0 = \delta S &= \delta\!\! \int \ell(u, n, v)\,dt
= \delta \!\!\int \ell_{B}(u)+\ell_{K d V}(n, v)\,dt
\\&
=\frac12\delta\!\! \int_{\mathcal{D}} u^2 
+  \nu n_x+2 n_x^3 - \gamma n_{x x}^2 \,d x \,dt    
\label{eq: HP var appendix}
\end{split}
\end{align}
with constant $\gamma$ and Eulerian fluid variables
\begin{align}
u = \dot{g}\,g^{-1}
\,,\quad
n = \phi\,g^{-1}
\,,\quad
\nu = \phi_t\,g^{-1}
\,.   \label{eq:HP var defs appendix}
\end{align}
The variations of velocity vector field $u=\dot{g} g^{-1}$ and of 0-forms $n=\phi g^{-1}$ and $\nu= \phi_t g^{-1}$ with respect to $\eta = \delta g g^{-1}$ and $\xi = \delta \phi g^{-1}$ are calculated as follows,\\

\begin{equation*}
\begin{array}{lll}
\delta u = \delta(\dot{g} g^{-1}) & \delta n = \delta (\phi g^{-1}) & \delta \nu = \delta (\phi_t g^{-1}) \\
= (\delta g)_t g^{-1} + \dot{g} \delta g^{-1} 
& = \delta \phi g^{-1} - \phi \delta g^{-1} 
& = (\delta \phi_t) g^{-1} - \phi_t \delta g^{-1} \\
= (\delta g g^{-1})_t - \delta g (g^{-1})_t - \dot{g} g^{-1} \delta g g^{-1} 
& = \delta \phi g^{-1} - \phi g^{-1} \delta g g^{-1} 
& = (\delta \phi g^{-1})_t + \delta \phi g^{-1} g_t g^{-1} - \phi_t g^{-1} \delta g g^{-1} \\
= \eta_t + \eta u - u \eta 
& = \xi - n \eta 
& = \xi_t + \xi u - \nu \eta \\
= \eta_t - \mathcal{L}_u \eta 
& = \xi - \mathcal{L}_\eta n 
& = \xi_t + \mathcal{L}_u \xi - \mathcal{L}_\eta \nu
\end{array}
\end{equation*}

Hence, we arrive at the Euler-Poincare equations:
\begin{equation}
    \begin{aligned}
\left(\partial_t+\operatorname{ad}_u^*\right) \frac{\delta l}{\delta u} & =\frac{\delta l}{\delta n} \diamond n+\frac{\delta l}{\delta \nu} \diamond \nu 
\,,\\
\left(\partial_t+\mathcal{L}_u\right) \frac{\delta l}{\delta \nu} & =\frac{\delta l}{\delta n} 
\,,\\
\left(\partial_t+\mathcal{L}_u\right) n & =\nu
\,.\end{aligned}
\label{EP-eqns-CoM}
\end{equation}
In 1D the diamond operator is given by $\diamond=-\partial_x$, then we have the variations of the Lagrangian in \eqref{eq: HP var appendix} as: 

\begin{equation}
    \begin{aligned}
\frac{\delta l}{\delta u} & =u \mathrm{~d} x \\
\frac{\delta l}{\delta \nu} & =\frac{n_x}{2} \mathrm{~d} x=\frac{1}{2} \mathrm{~d} n \\
\frac{\delta l}{\delta n} & =-\partial_x \frac{\delta l}{\delta n_x}+\partial_{x x} \frac{\delta l}{\delta n_{x x}} =\left(-\frac{1}{2} \nu_x-6 n_x n_{x x}-\gamma n_{x x x x}\right) \mathrm{d} x \\
& =\left(-\frac{1}{2} \partial_x\left(\left(\partial_t+u \partial_x\right) n\right)-6 n_x n_{x x}-\gamma n_{x x x x}\right) \mathrm{d} x \\
&=-\frac{1}{2} \left(\partial_t+u \partial_x\right) dn-\left(6 n_x n_{x x}+\gamma n_{x x x x}\right) \mathrm{d} x
\end{aligned}
\label{BSKdV-varsCoM}
\end{equation}
The second Euler--Poincar\'e equation in \eqref{EP-eqns-CoM} and the second variational equation in \eqref{BSKdV-varsCoM} combine to yield the following:
\begin{align}
\begin{split}
    \left(\partial_t+\mathcal{L}_u\right) \frac{\delta l}{\delta \nu} &=\left(\partial_t+\mathcal{L}_u\right)\left(\frac{1}{2} \mathrm{~d} n\right)
   \\=\frac{\delta l}{\delta n}& = - \frac{1}{2}\left(\partial_t+\mathcal{L}_u\right)  dn
    -6 n_x n_{x x} \mathrm{~d} x 
    -\gamma n_{x x x x} \mathrm{~d} x 
\end{split}
\end{align}
after commuting the differential and the Lie derivative in the above equation one finds: 
\begin{equation}
    d\left[ ({{\partial _t} + {{\cal L}_u}}) n + 3{n_x}^2
    + \gamma {n_{xxx}}\right] = 0
    \,,
\end{equation}
or equivalently,
\begin{equation}
    n_{x t}+u_x n_x+u n_{x x}+6 n_x n_{x x}+\gamma n_{x x x x} = 0
\,.\end{equation}
Now defining $v=n_x$ implies the second equation in Burgers–swept KdV system:
\begin{equation}
    v_t+(uv)_x=-(6 v v_x+\gamma v_{x x x})
\,.\end{equation}
In 1D the diamond operator is given by $\diamond=-\partial_x$, then we have 
\begin{equation}
    \left(\partial_t+3 u \partial_x\right) u=-\frac{n_x}{2} \nu_x+\left(\frac{1}{2} \nu_x-6 n_x n_{x x}-n_{x x x x}\right) n_x
    \,,
\end{equation}
which is the first equation in Burgers–swept KdV system:
\begin{equation}
    u_t + 3uu_x = -v(6vv_x + \gamma v_{xxx})
    \,.
\end{equation}

\newpage
\section{Spacetime Diagram for the examples} \label{appendix: spacetime}

\begin{figure}[H]
    \centering
    \includegraphics[width=0.9\linewidth]{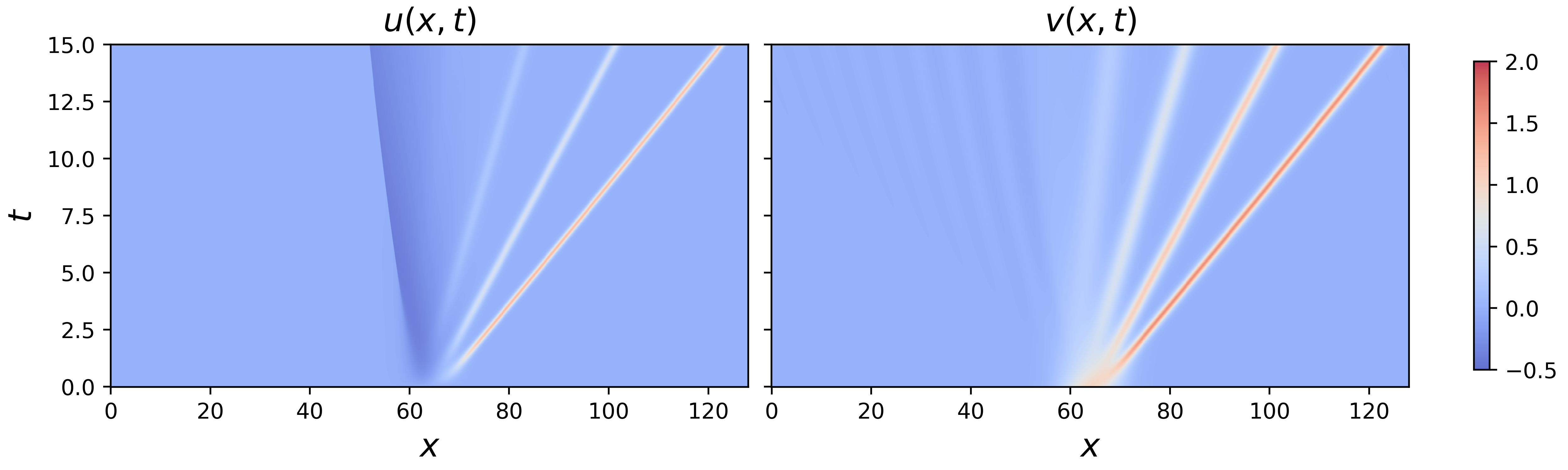}
    \caption{Spacetime diagram for emergence of soliton train (the bottom row of figure \ref{fig: kdv_emerge}). The initialization $u(x)=0, v(x)=\exp (-x^2 / 32)$. }
        \label{fig: st_gaussian}
\end{figure}

\begin{figure}[H]
        \centering
        \includegraphics[width=.9\linewidth]{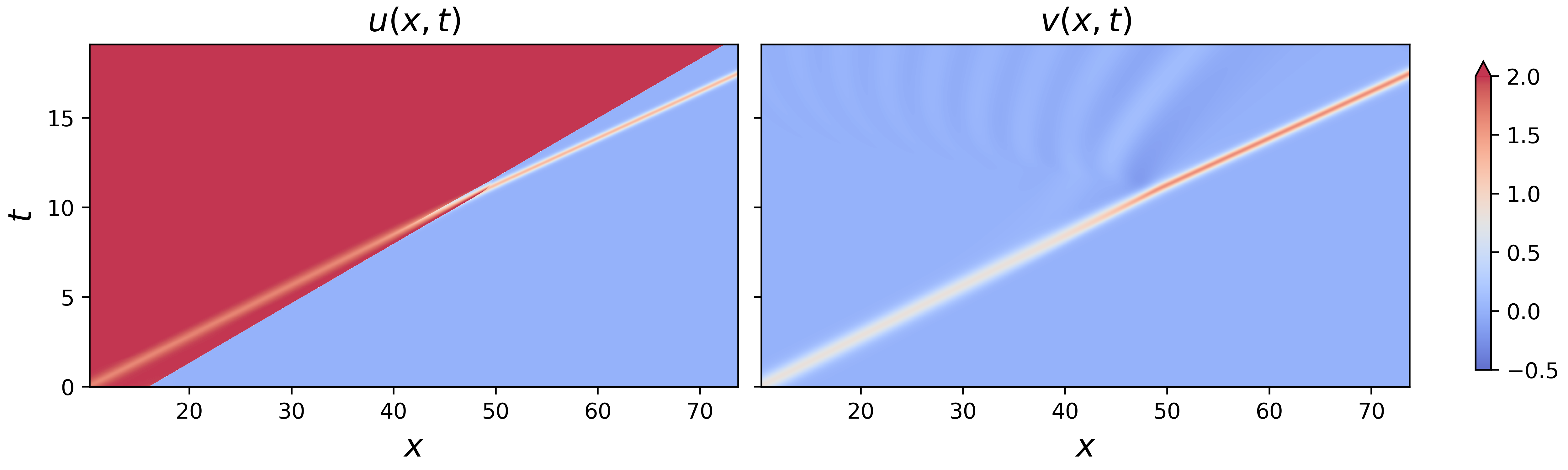}
         \caption{Spacetime diagram for soliton refraction from solitary wave to compound soliton (The bottom row of the figure \ref{fig: refraction}).}
        \label{fig: st_refraction_inner}
    \end{figure}
\begin{figure}[H]
    \centering
    \includegraphics[width=.9\linewidth]{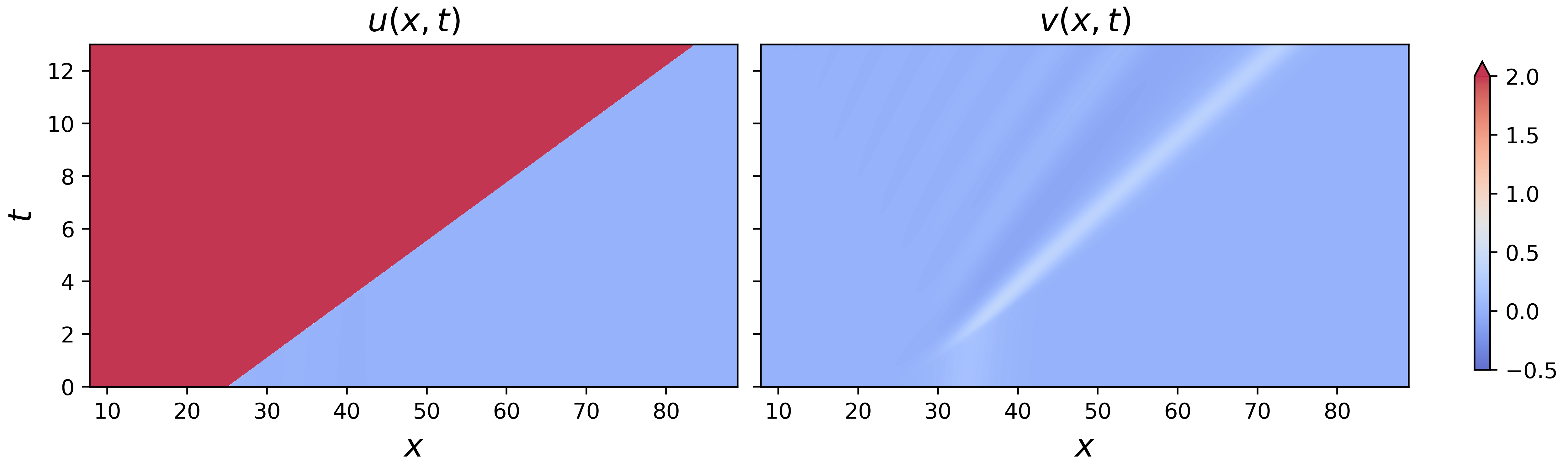}
    \caption{Spacetime diagram for soliton refraction (The top row of the figure \ref{fig: refraction}).}
        \label{fig: st_refraction}
\end{figure}

\begin{figure}[H]
    \centering
    \includegraphics[width=.9\linewidth]{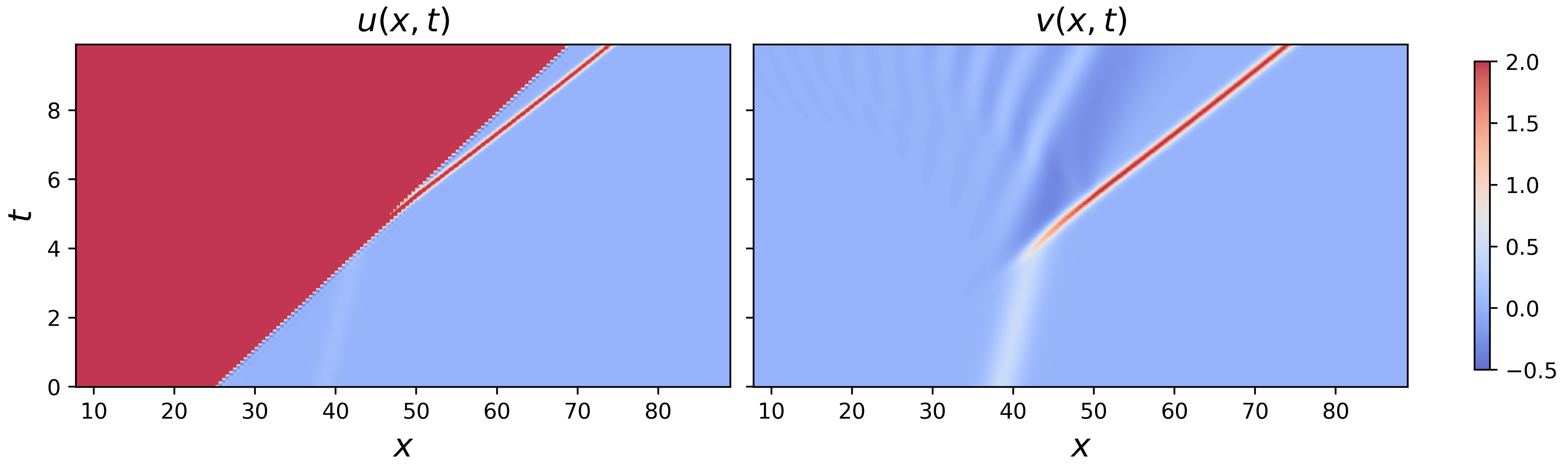}
    \caption{Spacetime diagram for soliton reflection (figure \ref{fig:Reflection}).}
    \label{fig: st_reflection}
\end{figure}

\begin{figure}[H]
        \centering
        \includegraphics[width=.9\linewidth]{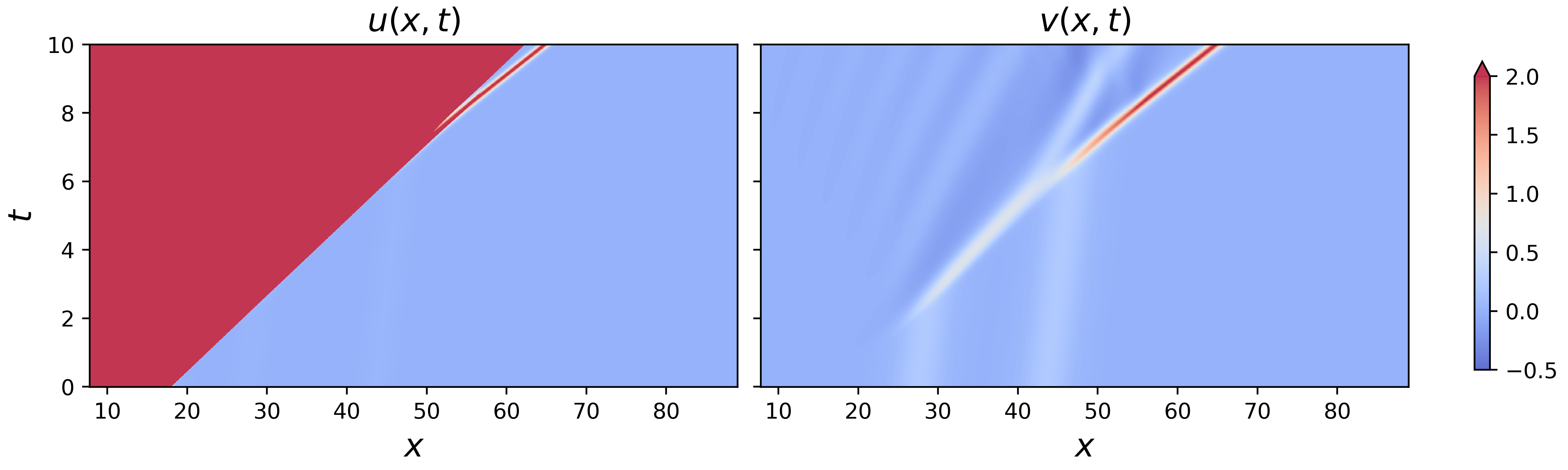}
        \caption{Spacetime diagram for soliton fusion (figure \ref{fig:fusion})}
        \label{fig: st_fusion}
    \end{figure}

\end{document}